\documentclass{article}
\usepackage{arxiv}
\usepackage{amsmath,amsfonts,amssymb,tikz,amsthm} 

\usepackage[ruled,vlined]{algorithm2e}
\usepackage{hyperref}

\newcommand{\R}{\mathbb{R}}
\newcommand{\An}{\mathcal{A}_n}
\DeclareMathOperator{\mean}{mean}
\renewcommand{\b}{\mathbf}
\newcommand{\matlab}{{\sc MATLAB~}}

\newtheorem{theorem}{Theorem}[section]
\newtheorem{lemma}{Lemma}[theorem]
\newtheorem{definition}{Definition}[theorem]

\newtheorem{corollary}{Corollary}[theorem]

\usepackage{booktabs}       

\usepackage{listings}
\definecolor{codegreen}{rgb}{0,0.6,0}
\definecolor{codegray}{rgb}{0.5,0.5,0.5}
\definecolor{codepurple}{rgb}{0.58,0,0.82}
\definecolor{backcolour}{rgb}{1,1,1} 
\lstdefinestyle{mycodestyle}{
    backgroundcolor=\color{backcolour},   
    commentstyle=\color{codegreen},
    keywordstyle=\color{magenta},
    numberstyle=\tiny\color{codegray},
    stringstyle=\color{codepurple},
    basicstyle=\ttfamily\footnotesize,
    breakatwhitespace=false,         
    breaklines=true,                 
    captionpos=b,                    
    keepspaces=true,                 
    numbers=none,                    
    numbersep=15pt,                  
    showspaces=false,                
    showstringspaces=false,
    showtabs=false,                  
    tabsize=2
}
\title{Generating large scale-free networks
with the Chung--Lu random graph model\thanks{This is the peer reviewed version of the article by Dario Fasino, Arianna Tonetto, and Francesco Tudisco \textit{``Generating large scale-free networks with the Chung–Lu
random graph model''}, which has been published in final form at \url{http://dx.doi.org/10.1002/net.22012}. This article may be used for non-commercial
purposes in accordance with Wiley Terms and Conditions for Use of Self-Archived Versions \url{http://www.wileyauthors.com/self-archiving}.}}

\author{%
\textbf{Dario Fasino, \ Arianna Tonetto} \\
Department of Mathematics, Computer Science and Physics,\\
University of Udine,\\
33100 Udine, Italy\\
\texttt{dario.fasino@uniud.it} 
\And
\textbf{Francesco Tudisco}   \\
  School of Mathematics\\
  Gran Sasso Science Institute\\
  67100, L'Aquila (Italy) \\
  \texttt{francesco.tudisco@gssi.it} 
}

\begin{document}

\maketitle

\begin{abstract}
Random graph models 
are a recurring tool-of-the-trade for studying 
network structural properties and benchmarking community detection  and other network algorithms.
Moreover, they serve as test-bed generators for studying
diffusion and routing processes on networks. 
In this paper, we illustrate how to generate large random graphs having a power-law 
degree distribution using the Chung--Lu model.
In particular, we are concerned with the fulfillment of a fundamental  hypothesis
that must be placed on the model parameters,
without which the generated graphs lose all the theoretical properties of the model,
notably, the controllability of the expected node degrees and the absence of 
correlations between the degrees of two nodes joined by an edge.
We provide explicit formulas for the model parameters
to generate random graphs that 
have several desirable properties, including a power-law degree distribution with any 
exponent larger than $2$,
a prescribed asymptotic behavior of the largest and average expected degrees, and the presence of a giant component.

\medskip

\end{abstract}


\section{Introduction}

A problem of central importance in network analysis is to be able to produce random graphs that resemble certain fundamental structural properties that emerge from the empirical observation of real-world networks. 
This problem is not only theoretically interesting, but also of practical relevance.

Random graph models are  useful tools 
for studying dynamical processes on networks, such as rumor or epidemic spreading on social networks \cite{BERTOTTI201655,MR1919737,Epidemics}. Moreover, they are widely used as test-bed data generators resembling desired properties from real-world sparse networks for benchmarking network optimization algorithms, including 
decentralized search and routing algorithms \cite{MR2275717}, as well as 
methods for clustering and partitioning \cite{fasino2018expected,mercado2016clustering}, and for the identification of mesoscale network structures such as communities \cite{bazzi2020framework,KaNe11} or core-periphery \cite{core-periphery,FRsymmetry}.

Another practical task where random graph generators are very helpful is to assess whether an emergent feature, e.g., a pattern or a quantitative property of a real-world network, is likely to be the result of self-organization or just a  chance.
In that case, a common practice is inspired by the statistical hypothesis testing,  where one compares measurements on real-world data with data sampled from a suitably chosen  random network generator. 
For example,  \cite{MR2736090}
presents an extensive 
and detailed picture of the emergence of community structures in large social and information networks, where empirical results are compared with both analytical results and stochastic simulations on a wide range of commonly used network models.  
Analogous comparisons between empirical data and statistical analysis have been carried out to uncover the non-random occurrence of certain small subgraphs, called motifs, as basic building blocks of many complex networks \cite{Milo2002Motifs}.

In this work we propose and analyze a random graph generator for graphs having a power-law degree distribution and a prescribed expected degree sequence.
A graph or network is said 
to have a power-law degree distribution with exponent $\gamma$ if the number $n_k$ of nodes having degree $k$ can be approximated by $\alpha k^{-\gamma}$, for some coefficient $\alpha$ that 
depends on the size of the graph.
Networks having a power-law degree distribution are called {\em scale-free}, due to the fact that 
the power law $P(x) = \alpha x^{-\gamma}$ 
fulfills the identity 
$P(c x) = C P(x)$, for some constant $C$ depending on $c$ but independent on $x$, 
so that the 
functional form 
remains unchanged under rescaling of the independent variable, apart from a multiplicative factor.

Power-law degree distributions have been observed in many real systems such as the Internet and the World-Wide-Web, food-webs, protein, and neural networks. In fact,
starting from the seminal work by Barab\'asi and Albert \cite{BA1}, a wealth of empirical studies has shown that 
the node degrees of many interesting real-world networks
follow a power-law distributions with exponent $2 < \gamma < 3$
\cite{BA2,PowLawSIREV10,FFF99,NeSIREV03}.
In these networks a small but not negligible fraction
of the nodes has very large degree. These nodes are called hubs.
On the other hand,   real-world networks 
are usually  {\em sparse}, that is,
the average degree is much smaller than the 
size of the network. Furthermore,  
even though networks  often undergo an evolution, growing in  size during time due to the addition of new nodes and edges, 
the average degree remains roughly constant, 
see e.g.,  
\cite{FFF99}. 
Both empirical and analytical considerations on the 
power law show that only when the exponent $\gamma$ lies in between $2$ and $3$
a growing network can have hubs and be sparse at the same time. 
More precisely, if we let the network size grow unbounded, depending on the value of $\gamma$ we have different behaviors: if $0<\gamma < 2$ then the average degree diverges and the network cannot be sparse, while if $\gamma>3$ then the degree variance is bounded and no large hub can appear \cite[\S 4.4]{BA2}.

The exponent $\gamma$ has a far-reaching impact also on spectral and  structural properties of scale-free networks. For example, 
it is known that when $\gamma > 5/2$ the principal eigenvector of the adjacency matrix of a scale-free network can
be localized, that is, a large fraction of its weight can be concentrated in a few entries while the remaining ones have zero (or close to zero) values \cite{EigLocalization}.
Localization effects can significantly diminish the effectiveness of spectral methods to quantify the importance of vertices and to detect communities in networks.
Moreover, if $\gamma > 7/3$ then the clustering coefficient, 
which is a measure of the density of triangles in a network, tends to zero as the network becomes large \cite{NeSIREV03}.
The incidence of the exponent $\gamma$ on the behaviour of an innovation diffusion model on scale-free networks has been analyzed in
\cite{BERTOTTI201655}, using both analytical and simulation techniques.

One of the earliest and most used generative models for scale-free networks is the Barab\'asi--Albert model \cite{BA1}. 
This model generates graphs that evolve in time 
and is easily described by a simple node-level self-organization rule, the so-called preferential attachment. The generation process is initiated from a  small subgraph, whose precise structure is asymptotically not influential on the final degree distribution. At each step of the generation, a new node is added to the network and is connected to $k$ pre-existing nodes, where $k$ is a fixed integer. 
Such $k$ nodes are chosen with a probability that is proportional to their current degree.  
With this rule, the degree distribution asymptotically follows a power law with exponent $\gamma = 3$, where the average degree is $k$ and the largest degree grows as $\mathcal{O}(\sqrt{n})$ on average, for a network with $n$ nodes.

Many other generative models for random scale-free networks have been developed since then. For example, 
the original preferential attachment model has been generalized along many directions, including the introducing of node deletion \cite{MoGhNe06}, node attractivity \cite{DoMeSa00}
and more general attachment rules \cite{krapivsky2,krapivsky}. Notably, the generative model introduced by Dorogovtsev, Mendes and Samukhin
in \cite{DoMeSa00} depends on a parameter that can be tuned to adjust the 
asymptotic 
degree distribution into 
any power law with exponent $\gamma \geq 2$. These models provide a justification for the emergence of power-law degree distributions for  preferential attachment growth processes. Moreover, 
they allow us to predict the behavior of a number of quantitative properties of large scale-free networks, and can be used as a baseline to detect 
deviations from this paradigm in real-world networks.
For example, important conclusions 
on the presence of statistical correlations between degrees of neighboring nodes 
in power-law graphs are shown in \cite{LitHof13} and rely on computational experiments with both random and real-world scale-free networks.
However, 
no generative model fits the ever-changing needs of  network analysis. 
Some models may introduce (or miss) certain structural properties as degree correlations,
node clustering or the occurrence of certain subgraphs. Thus, depending on the application and the context, one model can be preferred over another.

In this work, we focus on a random graph model originally proposed by Chung and Lu in \cite{CL02,chung,CL03}
and further thoroughly analyzed in \cite{chunglubook}. This model, which we refer to as the Chung--Lu model, 
is very flexible and conceptually very simple. 
The model depends on a parameter vector $\b w$ whose elements $w_1,\ldots, w_n$ set out  the node degrees in expectation, under suitable hypotheses. 
In other words, given a sequence $w_1,\ldots, w_n$, the model generates random graphs with $n$ nodes whose expected degrees  are exactly $w_1,\ldots,w_n$.
Our main goal is to explore to what extent large scale-free networks can be generated by means of the Chung--Lu random graph model. 
In particular, we show that, with a suitable choice of the  parameters defining the model, one can generate random graphs that have a number of desired properties. 
For example, they can have expected degrees having a power-law distribution with a prescribed exponent $\gamma$,
they can have specified average and largest expected degrees, allowing for graphs that are both sparse and have hubs, and they can have a giant component, i.e., a connected subgraph with a number of nodes that scales linearly with the number of nodes of the entire network. 
Notably, unlike the Barab\'asi--Albert model and its many variants, 
the generation of scale-free random graphs in the Chung--Lu model does not proceed incrementally, by adding one node after another, and the power-law degree distribution does not arise from a preferential attachment criterion.

The paper is organized as follows. In the next subsection we collect some preliminaries. 
Then, in Section \ref{sec:chunglu}  we present the main features of the Chung--Lu model.
In  Section 3  we discuss the main results of this work. 
First, we prove that a naive selection of the vector $\b w$ is unable to respect a basic assumption of the model when the size of the graph becomes large.
Then, we provide explicit expressions for parameters that allow us to generate sequences of power-law graphs with exponent $\gamma > 2$  
having a prescribed average degree, in expectation.
As a consequence, we demonstrate the possibility of generating large scale-free graphs having almost surely a connected component comprising a significant fraction of all nodes.
Finally, in Section \ref{sec:experiments}
we propose a simple and efficient procedure to generate random graphs from this model. We provide the \matlab  code of the graph generator and showcase its computing time performance. 
With the help of that generator, we perform various experiments to illustrate the theoretical results.

\subsection{Notations and basic results}

We use standard graph theoretical notations and definitions, cf.\ \cite{BA2,chunglubook}. 
A graph or network $G$ consists of a finite set $V$ of nodes or vertices and 
a set $E$ of edges. For notational simplicity, we identify $V$ with $\{1,2,\ldots,n\}$. Each edge $e\in E$ is an unordered pair of nodes.
If $e = \{i,j\}$ then we say that nodes $i$ and $j$ are adjacent,
and that $e$ is incident to $i$ (and also to $j$). An edge having the form $\{i,i\}$ is called loop. 
The adjacency matrix of $G$ is the $n\times n$ matrix $A$ such that $A_{ij} = 1$ iff $\{i,j\}\in E$ and $A_{ij} = 0$ otherwise.
The degree of a node $i\in V$ is the number of edges incident to $i$, denoted by $d_i$.
A loop contributes as one edge to the node degree.
If $d_i = 0$ then we say that node $i$ is isolated.
The volume of a subset $S\subseteq V$ is the number $\sum_{i\in S} d_i$.
For a vector $\b w=(w_1,\dots,w_n)$ we write $\mean(\b w)$ and $\mean_2(\b w)$ to denote its first and  second order means, $\mean(\b w) = \frac 1 n \sum_{i=1}^n w_i$ and $\mean_2(\b w) = \sum_{i=1}^n w_i^2/\sum_{i=1}^n w_i$, respectively.

In the sequel, we will repeatedly use the following elementary result:
If $f$ is a non-increasing function, then 
\begin{equation}   \label{eq:bounds}
    \int_{i_0+1}^{n+i_0+1}f(x)\,\mathrm{d} x \leq
 \sum_{i=1}^{n}f(i_0+i)
    \leq f(i_0+1) + \int_{i_0+1}^{n+i_0}f(x)\,\mathrm{d} x .
\end{equation}
Finally, the notation $f(n)\approx g(n)$ means that $f(n)/ g(n)\to 1$ as $n\to\infty$.

\section{The Chung--Lu random graph model}  
\label{sec:chunglu}

The Chung--Lu random graph model 
is one of the most widespread random graph models.
A detailed description and analysis of such model is presented in \cite{chunglubook}, which collects  important earlier
results concerning the number and size of connected
components \cite{CL02,CL06}, average distance and diameter \cite{chungdiam,chung}, 
and spectral properties of the associated adjacency matrix \cite{CL03}.
An independent appearance of the same model can be found in \cite{PrHi06}, where it has been introduced to analyze connectivities in protein-protein interaction networks. 
In \cite{ACL01} Aiello, Chung and Lu described an earlier random graph model for power law degree distributions which is equivalent asymptotically to the Chung--Lu model.

The Chung--Lu model is usually denoted by $G(\b w)$ where $\b w = (w_1,\ldots,w_n)^T$ is a vector of 
nonnegative real numbers which define the model. 
We can think of $w_i$ as a weight placed on node $i$ that determines its ability to generate edges. Indeed, 
a random graph $G\in G(\b w)$ is a graph with $n$ nodes, whose 
edges are generated independently from one another according to the following rule:
For $i,j = 1,\ldots,n$,
the probability of having an edge between nodes $i$ and $j$ is
\begin{equation}   \label{eq:defpij}
   p_{ij} = \frac{w_iw_j}{\sigma} , \qquad
   \sigma = \sum_{k=1}^n w_k .
\end{equation}
For mathematical convenience, loops are usually allowed. Hence, the expected number of edges 
incident to node $i$, that is, the expected degree of node $i$ is
\begin{equation}   \label{eq:pij}
   \sum_{j=1}^n p_{ij} = \frac{w_i}{\sigma} \sum_{j=1}^n w_{j} = w_i .
\end{equation}
In other words,  the expected degree of each node in a graph $G\in G(\b w)$ is equal to the corresponding coefficient in the vector $\b w$. 
As noted in \cite{Ne06}, the Chung--Lu model is the only random graph model 
where the probability of having an edge between nodes $i$ and $j$ is the product $g(w_i)g(w_j)$ of separate functions of the expected degrees of nodes $i$ and $j$.
As a consequence, this is the only random graph model that
does not introduce correlations among the degrees of the nodes joined by an edge. 
For that reason, the Chung--Lu model represents the fundamental ``null model'' for finding community structures in networks
\cite{FTsimax14,FTlaa18,Ne06}, since tightly interconnected node sets can be revealed as deviations from an uncorrelated random graph. 
Recall that
the existence of nontrivial correlations among such degrees in Barab\'asi--Albert networks is well known, see e.g.,
\cite{krapivsky2,LitHof13} or
\cite[Sect.\ 7.2]{NeSIREV03}.

Finally, we mention that the Chung--Lu model is the basic building block of more recent and advanced random graph models, such as the degree corrected stochastic block model
\cite{fasino2018expected,KaNe11}
and the block two-level Erd\H{o}s--R\'enyi model \cite{BTER}, which have been proposed as models for random graphs having a prescribed average degree distribution and
integrating clustering effects and community structures that appear in social networks.

\subsection{Admissible expected degree sequences}

For certain vectors $\b w$ the number $p_{ij}$ 
defined in \eqref{eq:defpij} can exceed $1$. 
In order to preserve its probabilistic meaning,
this issue is usually avoided by specifying the constraint
$\max_{i=1\ldots n} w_i^2 \leq \sigma$
on $\b w$, in such a way that $0< p_{ij}\leq 1$.
Other choices are possible; for example, various authors set $p_{ij}$ as $\min\{w_iw_j/\sigma,1\}$ 
or $w_iw_i/(\sigma + w_iw_j)$, see e.g., \cite{MiHa11}
or \cite[\S 6.2]{RGCN},
but in that case the identity  \eqref{eq:pij} is no longer valid
and all the interesting theoretical properties of the model are lost.
For this reason, in the present work we adopt the following definition.

\begin{definition}
Let $\b w\in\R^n$. We say that $\b w$ is {\em admissible}
if $w_1\geq w_2 \geq\ldots\geq w_n \geq 0$ and
$$
   w_1^2 \leq \sum_{k=1}^n w_k .
$$
Moreover, we denote by $\An$ the set of all admissible $n$-vectors.
\end{definition}

We point out that our use of the term ``admissible'' is different from that in \cite{chung}.

\section{Scale-free random graphs in the Chung--Lu model}

Owing to the pervasiveness of scale-free networks in the real world, 
and the fact that the Chung--Lu model allows us to choose in advance the expected degree distribution of random graphs, it is natural to ask whether or not it is possible to generate large 
power-law networks with arbitrary exponent from $G(\b w)$, by a suitable choice of the parameter $\b w\in\An$.
The main goal of this section is to answer that question.
In particular, we address the possibility of generating large graphs having prescribed statistical properties, such as the exponent of the power law and the average degree. 

Let $n_k\approx \alpha k^{-\gamma}$ be the degree 
profile
of a scale-free network $G$ with $n$ vertices and $\gamma > 1$,
that is, $n_k$ is the number of nodes having degree $k \geq 1$.
Then, for large $n$,
the number $N(k)$ of nodes with degree greater than or equal to $k$ 
can be approximated by 
\begin{equation}   \label{eq:n_k}
   N(k) = \sum_{i=k}^\infty n_i
   \approx \alpha\int_{k}^{+\infty}x^{-\gamma}\,\mathrm{d} x
   = \frac{\alpha}{\gamma-1}k^{1-\gamma} .
\end{equation}
For further reference, we note incidentally that the largest degree $d_{\max}$
for nodes in $G$
can be estimated by imposing $N(d_{\max}) = 1$, giving
$$ 
   d_{\max} \approx  
   \bigg( \frac{\alpha}{\gamma-1} \bigg)^{1/(\gamma-1)} .
$$ 
Moreover, assuming that the smallest degree in $G$ is $d_{\min}$ we have
$n = N(d_{\min}) \approx \alpha d_{\min}^{1-\gamma}/(\gamma-1)$, hence
\begin{equation}   \label{eq:kmax}
   d_{\max} \approx \bigg( \frac{n}{d_{\min}^{1-\gamma}} \bigg)^{1/(\gamma-1)} =
   d_{\min}\, n^{1/(\gamma-1)} .
\end{equation}
On the other hand, if $G\in G(\b w)$ with $\b w\in\An$ then it is reasonable to assume $N(w_i) \approx i$, since the expected degree of nodes $1,\ldots,i$ is 
greater than or equal to $w_i$.
Then, solving for $w_i$ the approximate identity
$i\approx \alpha w_i^{1-\gamma}/(\gamma - 1)$,   coming from \eqref{eq:n_k}, we obtain
\begin{equation}   \label{powerlaw}
   w_{i} \approx c
   i^{-\frac{1}{\gamma-1}} ,
   \qquad c = \biggl(
   \frac{\gamma-1}{\alpha}\biggr)^{\frac{1}{\gamma-1}} .
\end{equation}
This construction will be somewhat extended in the following result,
which is based on an idea found in \cite{chung}
and \cite[\S 5.7]{chunglubook}. 
Basically, a new parameter is introduced to shift the index $i$. 
As we will see later on, the presence of that parameter is very useful to overcome some limitations on the structure of the networks arising from \eqref{powerlaw}.
At the same time, that parameter may affect the number of nodes having small degree, resulting in a degree profile that partially deviates from the estimate \eqref{eq:kmax}. We will illustrate and discuss an example of this phenomenon in 
Section \ref{sec:experiments} (Figure \ref{fig:n_k_3}).

\begin{lemma}   \label{cl1}
Let $n\in \mathbb{N}$ and $\gamma>1$. Let $\b w\in\R^n$ be a vector such that 
\begin{equation}\label{powerlawattesa}
   w_i = c(i_0+i)^{-\frac{1}{\gamma -1}} ,
   \qquad i = 1,\ldots, n,
\end{equation} 
for some positive constant $c$ and $i_0 > -1$. If $\b w\in\An$ then a graph $G\in G(\b w)$ has an expected degree distribution that follows a power law with exponent $\gamma$. Namely, for $k\geq w_n$, the number of nodes with expected degree $k$ is approximately $\alpha k^{-\gamma}$ with $\alpha = (\gamma-1)c^{\gamma-1}$.
\end{lemma}

\begin{proof} 
As shown in \eqref{eq:pij}, the expected degree of the $i$-th vertex of $G$
is $w_i$. For $x \geq w_n$ let $N(x)$ be the number of nodes with expected degree greater than or equal to $x$.
Since the sequence \eqref{powerlawattesa} is strictly decreasing, we have $N(w_i) = i$.
Inverting the relation \eqref{powerlawattesa} we obtain 
$$
   i = \bigg(\frac{w_i}{c}\bigg)^{1-\gamma} - i_0 .
$$
Hence, for $k = 1,\ldots,n$ it holds
$$
   N(k) = \bigg(\frac{k}{c}\bigg)^{1-\gamma}-i_0 .
$$
Consequently, the number of nodes with expected degree $k$
is approximately
$$
   n_k = N(k) - N(k+1)
   = \frac{k^{1-\gamma} - (k+1)^{1-\gamma}}{c^{1-\gamma}} \approx 
   \frac{\gamma - 1}{c^{1-\gamma}} k^{-\gamma} ,
$$
where the last passage comes from the mean value theorem.
\end{proof}

The foregoing lemma does not ensure that 
the vector $\b w$ in \eqref{powerlawattesa} belongs to $\An$.
That condition can be met 
by a suitable choice of the constant $c$, as shown in the following result in the simplest case $i_0 = 0$.

\begin{theorem}   \label{esistenza} 
For all $\gamma > 1$ and for all $n\in\mathbb{N}$ let 
$w_i = ci^{-p}$ for $i = 1,\ldots, n$,
where $p = 1/(\gamma - 1)$, $0< c \leq c_{\max}$ and
$$
   c_{\max} = \begin{cases} 
   (1-(n+1)^{1-p})/(p-1) & \gamma \neq 2 \\
   \log(n+1) & \gamma = 2 . \end{cases}
$$
Then $\b w\in\An$.
\end{theorem}

\begin{proof} 
If $\gamma \neq 2$ then $p \neq 1$ and from the leftmost inequality in \eqref{eq:bounds} 
we have
$$
   w_1^2 = c^2 \leq c\frac{1 - (n+1)^{1-p}}{p-1} 
   = c\int_{1}^{n+1} x^{-p} \,\mathrm{d} x
   \leq c\sum_{i=1}^{n} i^{-p} = \sum_{i=1}^{n} w_{i} .
$$
Thus the claim follows from Lemma \ref{cl1}.
Analogously, when $\gamma=2$ we have 
$$ 
   w_1^2 = c^2\leq c\log (n+1)
   = c\int_{1}^{n+1}\frac{1}{x}\,\mathrm{d} x \leq c\sum_{i=1}^n \frac{1}{i}
   = \sum_{i=1}^{n} w_{i} ,
$$
so $\b w\in\An$ and the proof is complete.
\end{proof}

It is worth noting that the inequality $c \leq c_{\max}$
appearing in Theorem \ref{esistenza} is asymptotically tight.
Indeed, suppose that we set $c > (n^{1-p}-p)/(1-p)$ and 
$0 < p < 1$, that is, $\gamma > 2$. Then,
using the rightmost inequality in \eqref{eq:bounds} we obtain
$$
   w_1^2 > c \frac{n^{1-p}-p}{1-p}
   = c \Big( 1 + \int_{1}^{n} x^{-p} \,\mathrm{d} x \Big) \geq
   c\sum_{i=1}^{n} i^{-p} = \sum_{i=1}^{n} w_{i},
$$
hence in this case $\b w\notin\An$. 
When $\gamma = 2$ the analogous conclusion follows by setting $c > \log n + 1$,
and for $1<\gamma<2$ by assuming $c > 1/(p-1)$. 
In summary, we have the following result.

\begin{corollary}   \label{cor:necessity}
Let $p = 1/(\gamma - 1)$ and $w_i = c i^{-p}$ for $i = 1,\ldots,n$.
If $\b w\in\An$ then $c \leq \hat c$ where
$$
   \hat c = \begin{cases} 
   1/(p-1) & 1 < \gamma < 2 \\
   \log n+1 & \gamma = 2 \\
   (n^{1-p} -p)/(1-p) & \gamma > 2 . \end{cases}
$$
\end{corollary}

For simplicity, in the proof of Theorem \ref{esistenza} 
we have only shown that admissible vectors of the form \eqref{powerlawattesa} with $i_0=0$ exist. However, in a similar way, it is possible to show that there are Chung--Lu scale-free networks obtained from admissible vectors of the same form, for every $i_0 > -1$. Indeed, if $\gamma\neq 2$
then it is enough to choose $c$ and $i_0$ such that
\begin{equation}   \label{ammissibile}
   c \leq (1+i_0)^{2p} \frac{(n+i_0+1)^{1-p}-(1+i_0)^{1-p}}{1-p} .
\end{equation}
In fact, choosing $c$ as in \eqref{ammissibile} we obtain 
\begin{equation}   \label{eq:7.5}
    w_1^2 = c^2(1+i_0)^{-2p}
    \leq c\frac{(n+i_0+1)^{1-p}-(1+i_0)^{1-p}}{1-p}
    = \int_{1+i_0}^{n+i_0+1}cx^{-p}\,\mathrm{d}x
    \leq \sum_{i=1}^{n} w_i .
\end{equation}
For $\gamma = 2$ the condition analogous to \eqref{ammissibile} is
$$ 
    c\leq(i_0+1)^2\bigl(\log(n+i_0+1)-\log(i_0+1)\bigr) .
$$ 
The introduction of $i_0$ allows us to prescribe the average expected degree $d$ and the 
largest expected degree $m$ in a Chung--Lu scale-free graph, under appropriate hypotheses. 
The next result, which is based on an idea found in \cite{chung} and \cite[p.\ 109]{chunglubook}, explains how to achieve this goal by a suitable choice of the parameters $i_0$ and $c$ 
in \eqref{powerlawattesa}.

\begin{theorem}   \label{corollario}
Let $\gamma>2$. 
Suppose that $d = d(n)$ and $M = M(n)$ are two nondecreasing functions of $n \in \mathbb{N}$ such that $0<d(n)\leq M(n)\leq n$, 
\begin{equation}    \label{eq:limd/M}
   \lim_{n\to\infty} \frac{d(n)}{M(n)} = 0 ,
\end{equation}
and there exists a constant $0 < \eta < 1$ independent on $n$ such that 
\begin{equation}   \label{eq:condX}
      \eta n d(n) \geq M(n)^2 .
\end{equation}
For each $n \in \mathbb{N}$ and
$i = 1,\ldots n$ let $w_i=c(i_0+i)^{-p}$ where $p=1/(\gamma-1)$,
\begin{equation}   \label{eq:c_i_0}
   c =c(n)= (1-p)d(n) n^p , \qquad    
   i_0 = i_0(n) = 
   n\bigg(\frac{(1-p)d(n)}{M(n)} \bigg)^{1/p} - 1 . 
\end{equation}
Then, \begin{enumerate}
    \item for $n$ sufficiently large it holds $\b w\in\An$
    \item any $G\in G(\b w)$ has an expected degree distribution that follows a power law with exponent $\gamma$, that is, for $k\geq w_n$, the number of nodes with expected degree $k$ is 
    $n_k \approx \alpha k^{-\gamma}$ with $\alpha = (\gamma-1)c^{\gamma-1}$
    \item the largest expected degree of any $G\in G(\b w)$ is $M(n)$ and the average expected degree is asymptotically $d(n)$, in the sense that 
\begin{equation}   \label{eq:mean_d}
   \lim_{n\to\infty} \frac{\mean(\b w)}{d(n)} = 1 .
\end{equation}
\end{enumerate}
\end{theorem}

\begin{proof}
To keep the notation simple, we sometimes omit the explicit dependence on $n$ of $M$, $d$, and other variables to be introduced in the proof.
Firstly, note that \eqref{eq:c_i_0} implies that the largest expected degree is
\begin{equation*}
   w_1 = c(i_{0}+1)^{-p} = (1-p)d n^p
   \bigg(n\bigg(\frac{(1-p)d}{M}\bigg)^{1/p}\bigg)^{-p} = M .
\end{equation*}
By the 
rightmost formula in \eqref{eq:c_i_0}
we can write $i_0 + 1 = n\theta$ where $\theta=\theta(n)$ is a function of $n$ such that $\theta(n) > 0$ and $\lim_{n\to+\infty}\theta(n)=0$.
Define
$$
    \epsilon_n = \sum_{i=1}^n (i_0+i)^{-p} -
    \int_{i_0+1}^{n+i_0+1}x^{-p} \,\mathrm{d}x .
$$
From \eqref{eq:bounds} we have $\epsilon_n > 0$ and
$$
    \epsilon_n \leq 
    (i_0+1)^{-p} - \int_{n+i_0}^{n+i_0+1}x^{-p} \,\mathrm{d}x 
    < (i_0+1)^{-p} .
$$ 
In particular, $c\epsilon_n < M$. Hence, from \eqref{eq:condX}, we have
$c\epsilon_n/n d = \mathcal{O}(1/\sqrt{n d})$.
Moreover,
\begin{align*}
   \sum_{i=1}^n w_i =
   c\sum_{i=1}^n (i+i_0)^{-p} 
   & = c \bigg( \int_{i_0+1}^{n+i_0+1}x^{-p} \,\mathrm{d} x + \epsilon_n  \bigg)  \\
   & = \frac{c}{1-p}\Bigl((n+i_0+1)^{1-p}-(i_0+1)^{1-p} \Bigr) + c \epsilon_n \\
   & = dn \bigl((1+\theta)^{1-p}-\theta^{1-p}\bigr) + c \epsilon_n . 
\end{align*}
Hence,
\begin{equation}   \label{eq:averagedegree}
   \lim_{n\to\infty} \frac{1}{n d(n)}\sum_{i=1}^n w_i = 
   \lim_{n\to\infty} \bigl((1+\theta)^{1-p}-\theta^{1-p}\bigr) + 
   \frac{c\epsilon_n}{n d} = 1
\end{equation}
and we get \eqref{eq:mean_d}.
Finally, if \eqref{eq:condX} holds then
for sufficiently large $n$ we have 
$$
   \sum_{i=1}^n w_i \geq nd(1 -\theta^{1-p}) \geq
   nd\eta \geq M^2 = w_1^2 .
$$
Hence $\b w\in\An$ and the proof is complete.
\end{proof}

It is worth noting that the hypothesis 
\eqref{eq:limd/M}
is by no means restrictive. 
Indeed, 
in all scale-free networks the degree of hub nodes is far greater than the average degree,
and their ratios behaves as indicated in \eqref{eq:limd/M} when the network size diverges.
Moreover, the condition \eqref{eq:condX} is almost optimal. 
In fact,
the condition $\b w\in\An$ is equivalent to $M^2 = w_1^2 \leq \sum_i w_i = n\mean(\b w)$
where $\mean(\b w)= (\sum_i w_i)/n$ is the expected average degree of a random graph from $G(\b w)$.
Hence, if \eqref{eq:condX} is violated then no arbitrarily large admissible vector $\b w$ can be obtained. 
This conclusion agrees with the findings in \cite{cutoffs},
which show that the largest node degree in a growing network with no degree correlation and average degree $d_{\mathrm{avg}}$ must be $\mathcal{O}(\sqrt{nd_{\mathrm{avg}}})$,
independently on the degree distribution.

On the other hand, the condition \eqref{eq:condX} is quite stringent,
at least in some scenarios.
For example, if $d(n)$ is upper bounded by a constant then that condition implies that 
the largest expected degree
$M(n)$ must grow not faster than $\sqrt{n}$.
However, we observed in \eqref{eq:kmax} that 
if the degree distribution follows a power law then
the largest degree  behaves as $d_{\max} = \mathcal{O}(n^p)$. Hence,
the estimate \eqref{eq:kmax} can be attained 
for large $n$ only if $p \leq \frac12$, that is $\gamma \geq 3$.
Anyway, the hypotheses of Theorem \ref{corollario} are essentially the most general possible for having admissible expected degree sequences following a power law.

\subsection{Conditions for the giant component}

A very relevant element in the analysis of random graphs
is the presence of giant components. It is customary to say that a network has a giant (connected) component if there is a connected subgraph comprising a significant fraction of all the nodes. More precisely, in a network whose number
of nodes $n$ increases over time, a giant component is a maximal connected subgraph
whose size is $\mathcal{O}(n)$. Aiello, Chung and Lu \cite{ACL01} obtained very detailed results concerning existence, uniqueness and the size not only of the giant component, but also of the smaller connected subgraphs of random power-law graphs. In particular, they showed that a random power law graph with exponent $\gamma < \gamma_0 \approx 3.47875$ almost surely has a unique giant component.\footnote{
In probability theory, it is customary to say that a property depending on an integer $n$ holds almost surely 
if the probability that it holds tends to $1$ as $n$ goes to infinity.} 
However, 
the results in \cite{ACL01} are based on the fundamental assumption that
the largest degree in a random power-law graph with $n$ nodes and exponent $\gamma>1$
is roughly proportional to $n^{1/\gamma}$. That assumption is questionable. For example,
in \cite[\S4.3]{BA2} 
and \cite[\S 3.3.2]{NeSIREV03} 
it is argued that the largest degree behaves as $n^{1/(\gamma - 1)}$,
as we also discussed in \eqref{eq:kmax}.

More generally, for Chung--Lu random graphs with a generic expected degree sequence, a giant component appears if the expected average degree is larger than $1$. Indeed, the following result holds, see 
\cite{CL06} and \cite[Thm.\ 6.14]{chunglubook}.\footnote{
In \cite{CL06} and \cite[Thm.\ 6.14]{chunglubook}, the hypothesis $\b w\in\An$ is tacitly assumed.}

\begin{theorem}   \label{thm:gcc}
Let $\b w\in\An$, $\mean(\b w)= (\sum_i w_i)/n$ and $\mean_2(\b w) = (\sum_i w_i^2)/(\sum_i w_i)$. If $\mean(\b w)>1$ then almost surely a random graph $G\in G(\b w)$ has a unique giant component, 
whose overall weight is approximately 
$(\lambda_0 + o(1)) \sum_i w_i$,
where $\lambda_0$ is the unique positive root of the equation
$$
   \sum_i w_i \mathrm{e}^{-w_i \lambda} = (1-\lambda) \sum_i w_i .
$$
Moreover, if $\mean_2(\b w) < 1-\varepsilon$ for some $\varepsilon>0$ independent on $n$ then almost surely there is no giant component in $G$.
\end{theorem}

In the preceding claim, the overall weight of a node subset $S$ of a graph $G\in G(\b w)$ is the number $\sum_{i\in S}w_i$.
Note that
$\mean(\b w) \leq \mean_2(\b w)$ due to the Cauchy--Schwartz inequality.
Whilst nothing is known about the existence of a giant component when
$\mean(\b w) < 1 < \mean_2(\b w)$,
it is important to understand if it is possible to generate vectors $\b w \in \An$ 
for arbitrarily large $n$ and $\mean(\b w) > 1$.
For example, the definition 
considered in Theorem \ref{esistenza}, i.e.\ 
$w_i = ci^{-p}$,
severely constrains
the possibility of having a giant component, as shown in the following result.

\begin{corollary}    
Let $0 < p \leq 1$
and let the vector $\b w\in\An$ be defined as $w_{i} = ci^{-p}$ for $i=1,\ldots, n$. 
If $\mean(\b w) > 1$ then either $p \leq \frac{1}{2}$ or 
$\frac12 < p <1$ and 
$$ 
   n < 1/(1-p)^{\frac{2}{2p-1}} .
$$ 
\end{corollary}

\begin{proof} 
Firstly, we exclude the case $p = 1$. Indeed, in that case Corollary \ref{cor:necessity} implies
$c \leq \log n + 1$, and for $n > 1$ it holds
$$
   \mean(\b w) = \frac{c}{n} \sum_{i=1}^n i^{-1} 
   \leq \frac{(\log n + 1)^2}{n}  < 1 .
$$
For $0 < p < 1$, again by Corollary \ref{cor:necessity} we have 
$c < n^{1-p}/(1-p)$. Hence, 
$$
	\mean(\b w) = \frac{c}{n} \sum_{i=1}^n i^{-p} 
	< \frac{n^{-p}}{1-p} \bigg( 1 + \int_1^n x^{-p}\, \mathrm{d} x\bigg)
	< \frac{n^{1-2p}}{(1-p)^2} .
$$
Letting $\mean(\b w) > 1$ we obtain
\begin{equation}   \label{limiten}
	n^{1-2p} > (1-p)^2 .
\end{equation}
This inequality holds for every $n\in\mathbb{N}$ if and only if $p \leq \frac{1}{2}$. On the other hand,
if $p > \frac12$ then 
\eqref{limiten} yields the upper bound $n < (1-p)^{2/(1-2p)}$ and the theorem is proved.
\end{proof}

As shown in the preceding corollary, the construction of large scale-free random graphs having a giant component using the formula $w_i = c i^{-1/(\gamma-1)}$ is not always possible if $\gamma < 3$.
On the other hand, for $\gamma \geq 3$, 
we can obtain vectors $\b w\in\An$ with arbitrarily large $n$ and  such that the graphs in $G(\b w)$ have a prescribed expected average degree $d > 1$.
Indeed, let $p=1/(\gamma-1)$ and $c = d(1-p)n^p$. If $0<p< 1/2$, then for sufficiently large $n$ it holds
$$
   c = d(1-p)n^p \leq \frac{(n+1)^{1-p}-1}{1-p} ,
$$
so we have $\b w\in\An$ by Theorem \ref{esistenza}.
Moreover, using essentially the same argument as the one of the proof of Theorem \ref{corollario}, it is not difficult to observe that the expected average degree of a graph in $G(\b w)$ is asymptotically equal to $d$.
The same conclusions carry over the case $p = 1/2$, that is $\gamma = 3$, under the additional constraint $d\leq 4$.
It is worth noting that, with this construction, we have
$$
   \frac{w_1}{w_n} = (n+1)^p \, ,
$$
that is, the ratio between the expected largest and smallest degrees 
behaves as the analytical estimate obtained in \eqref{eq:kmax}.
Taking everything into consideration,
the next result provides a way to construct sequences of large scale-free Chung--Lu  random graphs with power law degree distribution $n_k=\alpha k^{-\gamma}$ and expected average degree larger than one, for every exponent $\gamma > 2$.
\begin{corollary}   \label{cor:M(n)}
Let $0 < p < 1$ and let $d$ be a fixed number larger than one.
For $n\in\mathbb{N}$ define 
\begin{equation}   \label{eq:M(n)}
   M(n) = \sqrt{dn/2} ,
\end{equation}
Moreover, for $i = 1,\ldots,n$
let $w_i = c(i_0 + i)^{-p}$ where 
$c$ and $i_0$ are defined as in \eqref{eq:c_i_0}.
Then for sufficiently large $n$ one has $\b w\in\An$ and
the expected average degree of a random graph from $G(\b w)$ is asymptotically $d$. 
\end{corollary}

\begin{proof}
The claim is a straightforward consequence of Theorem \ref{corollario}.
\end{proof}

In our computational examples, we always adopt the formula \eqref{eq:M(n)}, and all resulting vectors $\b w$ are admissible.
We conclude with the following example showing that, under certain restrictions, the proposed construction can 
produce random scale-free networks similar to those produced by the Barab\'asi--Albert model. 
Let $\b w\in\mathbb{R}^n$ be the vector with entries $w_i = ci^{-1/2}$ for $i = 1,\ldots,n$ where $0 < c \leq 2\sqrt{n+1}-2$.
By Theorem \ref{esistenza}, $\b w\in\An$ and a graph $G\in G(\b w)$ has an expected degree distribution following a power law with exponent $\gamma = 3$ and expected average degree $d = 2c/\sqrt{n}$. In particular, 
the ratio between the largest and the average degree is approximately $\sqrt{n}/2$, as in the Barab\'asi--Albert model.

\section{Computational examples}  \label{sec:experiments}

In what follows, we present a possible implementation 
of a random graph generator based on the Chung--Lu model.
Then, 
we show a number of examples to illustrate the analytical results in the previous sections and to demonstrate the flexibility 
of the proposed generator.

\subsection{An efficient generator of Chung--Lu random graphs}

Since real-world networks are often very large, the availability of efficient random graph generators is crucial for practical purposes.
The obvious algorithm based on \eqref{eq:defpij} considers each node pair and generates the corresponding edge according to the prescribed probability. 
This is the approach implemented in the function {\tt sticky}
of the \matlab package CONTEST \cite{contest}, which is probably the earliest implementation of a $G(\b w)$-type random graph generator.
The resulting computational cost is $\mathcal{O}(n^2)$ for a graph with $n$ nodes, which is unsuitable for large graphs.
An efficient implementation of the Chung--Lu model is included within 
the block two-level Erd\H{o}s--R\'enyi (BTER) algorithm \cite{BTER}.
That algorithm
has been designed with the goal of producing random graphs resembling certain social network properties,
which are assembled as the union of Erd\H{o}s--R\'enyi random graphs.
A more efficient algorithm, 
based on a rejection sampling technique, has been described in \cite{MiHa11}.
The average computational cost of that algorithm is essentially linear in the number of nodes and edges.

We propose here a new generator for graphs in $G(\b w)$. The main advantage of our generator is its 
compact, highly modular implementation, which runs very efficiently on \matlab and other high-level scientific programming languages,
due to the absence of explicit for-loops. 
Our algorithm is based on some ideas laid out in
\cite{BTER,Nuovissimo} and, as for the method proposed in \cite{MiHa11}, has a running time essentially proportional to the number of edges in the graph,
over a very large size range.

The algorithm implements the principle called ``ball dropping'' in \cite{Nuovissimo}.
According to such principle, the algorithm initially generates two random vectors, {\tt I} and {\tt J},
whose entries are node indices. In each vector, the node index $i\in\{1,\ldots,n\}$ appears $\ell w_i/\sum_jw_j$ times on average, 
where $\ell$ is the length of the random vectors.
Then, edges are generated by joining nodes {\tt I}$(i)$ and {\tt J}$(i)$ for $i = 1,\ldots,\ell$.  
An obvious choice would be to choose $\ell = \lceil m\rceil$ where $m = \frac12\sum_{i=1}^n w_i$ is the expected number of edges in a random graph from $G(\b w)$.
However,
the procedure may produce repeated edges, which can be removed at the 
end of the 
algorithm. To counteract the removal, the vector length 
is set to $\lceil m + e\rceil$ where 
$e = \frac12\mean_2(\b w)^2$
is an estimate on the number of multiple edges that are generated initially (assuming $\b w\in\An$, of course).
In this way, the number of generated edges is not constant, and its average value is 
about $m$.
The pseudo-code of the algorithm is shown in Algorithm 1.
The random vectors {\tt I} and {\tt J}
are computed by means of the auxiliary function {\tt sample}.

The computational complexity of the algorithm is dominated by that of the sorting of the random vector $r$ in {\tt sample}, whose length is 
essentially equal to
the number of edges in the network. 
All other tasks are linear in the number of nodes or edges. Thus, the computational complexity of Algorithm 1 for a network with $m$ edges is $\mathcal{O}(m \log m)$. However, the effective wall-clock time is almost linear in $m$ for sparse graphs and over a very large size range, as we show in the sequel. In fact, the algorithm can be implemented making extensive use of highly vectorizable low-level functions.

\begin{algorithm}[t]
\DontPrintSemicolon
\caption{Chung--Lu random graph generator}   \label{alg:1it}
\KwIn{$\b w = (w_1,\ldots,w_n)$, expected degree vector}
\KwOut{random graph $G(V,E)$ with $V = \{1,\ldots,n\}$ }
\SetKwFunction{sample}{sample}
\BlankLine
$wsum_0 \leftarrow 0$\;
\For{$i = 1$ \KwTo $n$}{
$wsum_{i} \leftarrow wsum_{i-1}+w_i$ \;
}
$wsum \leftarrow wsum / wsum_n$\;
$\ell \leftarrow \lceil \frac12 \sum_{i=1}^n w_i + \frac12 \mean_2(\b w)^2 \rceil$ \;
$I \leftarrow \sample{$\ell,wsum$}$ \;
$J \leftarrow \sample{$\ell,wsum$}$ \;
$E \leftarrow \emptyset$ \;
\For{$i = 1$ \KwTo $\ell$}
   {
   $E \leftarrow E \cup \{I(i),J(i)\}$ \;  
   }
\;
\SetKwProg{Fn}{Function}{}{}
\Fn{\sample{$\ell, s$}}
{
\KwIn{integer $\ell$, vector $s = (s_0,\ldots,s_n)$}
\KwOut{vector $bin = (bin_1,\ldots,bin_\ell)$}
\BlankLine
choose $r = (r_1,\ldots,r_\ell) \in (0,1)^\ell$ uniformly at random\;
sort $r$ in increasing order\;
$ k \leftarrow 0$ \;
\For{$i = 1$ \KwTo $\ell$}{
   \While{$r_i \geq s_k$}{
      $k \leftarrow k+1$ \;
   }
   $bin_i \leftarrow k$ \;
}        \KwRet $bin$\;
}
\end{algorithm}

Figure \ref{fig:code} shows the \matlab 
implementation of this algorithm, which is also available through the GitHub repository {\tt https://github.com/ftudisco/scalefreechunglu}
together with a {\sc Python} version. 
The auxiliary function {\tt sample} in Algorithm 1 is implemented by a 
one-line call to the \matlab builtin function {\tt discretize}.

\begin{figure}[t]
\centering
\lstset{language=matlab}
\begin{lstlisting}
function A = CL_generator(w)
% INPUT:    w = expected degrees vector
% OUTPUT:   A = binary adjacency matrix of a graph G in G(w)
    n = length(w);
    m = (dot(w,w)/sum(w))^2 + sum(w); 
    m = ceil(m/2);
    wsum = [0 ; cumsum(w(:))]; 
    wsum = wsum/wsum(end);
    I = discretize(rand(m,1),wsum);
    J = discretize(rand(m,1),wsum);
    A = sparse([I;J],[J;I],1,n,n);
    A = spones(A);
end
\end{lstlisting}
\caption{\matlab function that, given the vector of expected degrees $\b w=(w_1,\dots,w_n)$, generates the adjacency matrix of a random graph $G\in G(\b w)$.} \label{fig:code}
\end{figure}

\begin{figure}[t]
    \centering
    \includegraphics[width=.47\textwidth]{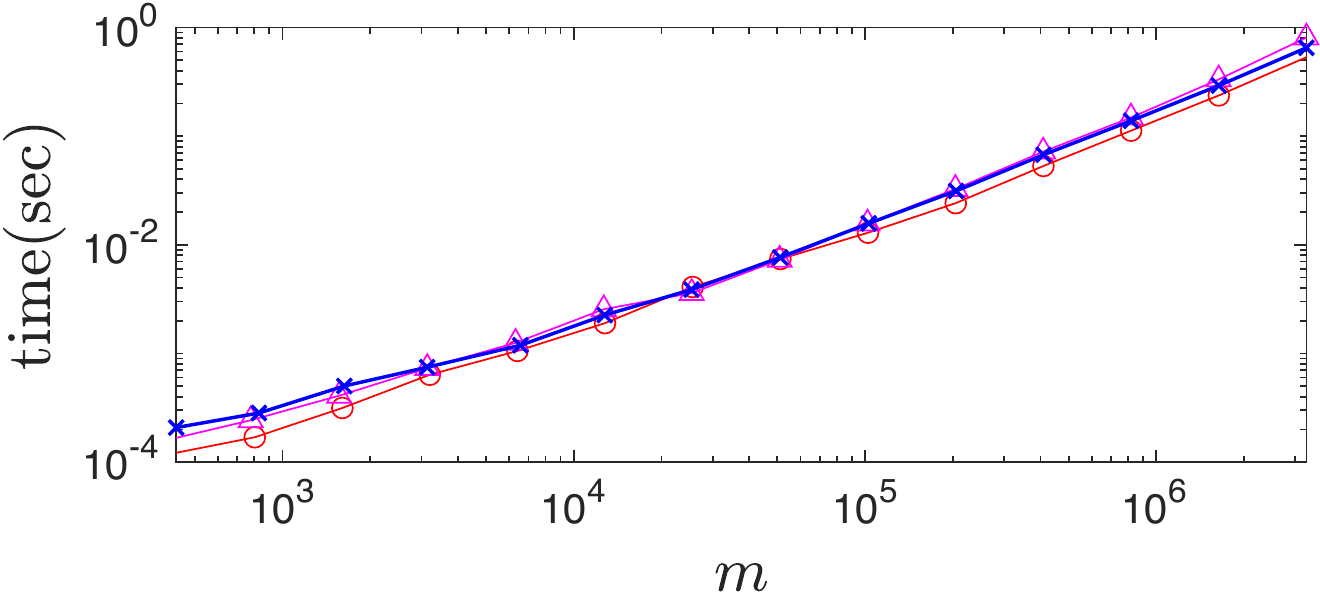}
    \hspace{5mm}
    \includegraphics[width=.47\textwidth]{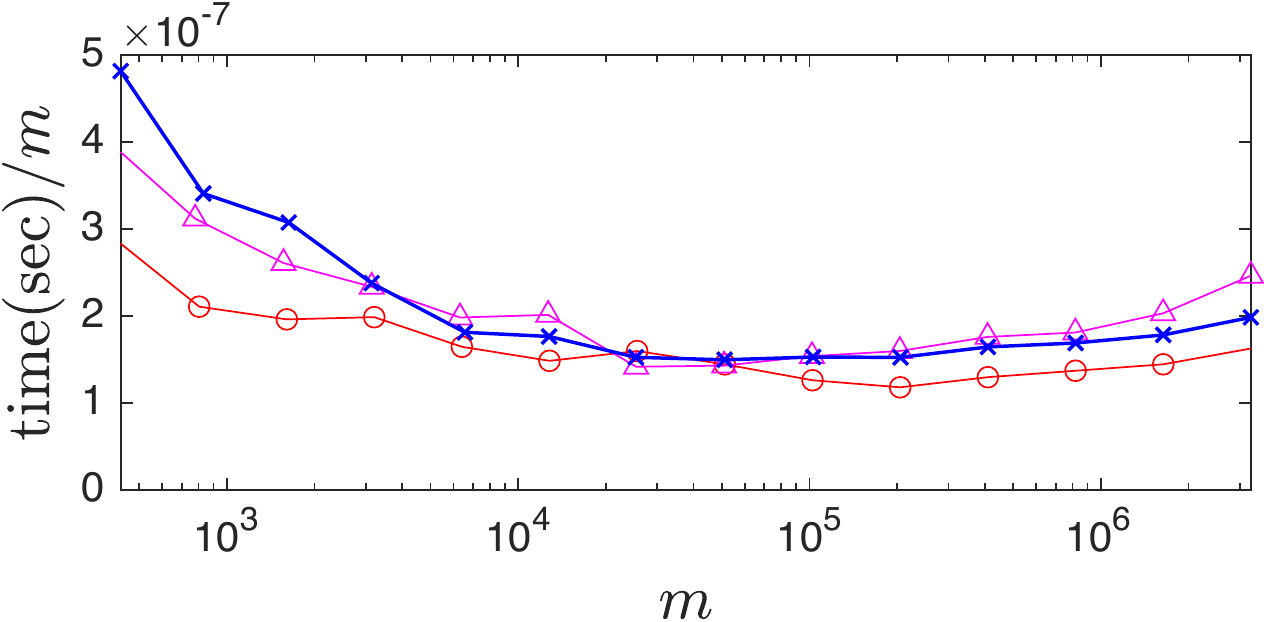}
    \caption{Mean performance of the random graph generator over 10 random trials and three different degree profiles. Each line represents a sequence of random graphs generated by a different 
expected 
degree sequence: constant $w_i = 4$ (red circles), uniform in $[0,8]$ (blue crosses), power law with $\gamma = 3$ and expected average degree $\mean(\b w) = 4$ (magenta triangles). The horizontal axes indicate the number $m$ of actual edges in graphs. Left: Running time vs number of edges. Right: Average time per edge.}
    \label{fig:timings}
\end{figure}

Figure \ref{fig:timings} demonstrates the running time of the algorithm for various random graphs in the $G(\b w)$ model. We performed a series of experiments in \matlab v.2019b on a laptop PC endowed with a i7-8550U processor and 1.80GHz CPU clock. 
We generated graphs with $n$ nodes where $n = 100 \cdot 2^{k-1}$ for $k = 1,\ldots,14$ and the vectors $\b w$ are chosen
using three different expected degree 
sequences: a constant sequence with 
$w_i = 4$, a random sequence where each $w_i$ is set by a pseudo-random generator uniform in $[0,8]$, and a power law distribution with $\gamma = 3$ and expected average degree $\mean(\b w) = 4$, obtained by the formula $w_i = 2\sqrt{n/i}$. 
In all experiments, the expected number of edges is $4n$, even though 
in practice that number varies, due to the randomness of the generator.
We denote by $m$ the actual number of edges in each graph.
In Figure \ref{fig:timings}, the $x$-axis corresponds to the number of edges in the graph,
illustrating that the execution time of the algorithm scales  essentially linearly with the number of edges.
The left picture shows the running time obtained by each experiment, while the right picture shows the average time per generated edge.  
The timings are the averages over $10$ runs per each dimension $n$ and degree profile. 

\subsection{Examples}

To illustrate the effectiveness of the construction devised in Theorem \ref{corollario}, hereafter
we show various statistics 
on random graphs built using Algorithm \ref{alg:1it} (as implemented in the code in Figure \ref{fig:code}) where we selected the expected degree sequence according to the formula in \eqref{eq:c_i_0}, with $n = 10000$, $d = 10$ and $M = \sqrt{dn/2}$, as per Corollary \ref{cor:M(n)}. The resulting vectors $\b w$ are admissible.
The other relevant parameters are reported in Table 1. 
The last column in that table reports the average degree $\frac1n\sum_{i=1}^nd_i$ averaged over $10$ realizations of $G\in G(\b w)$.

Figure \ref{fig:deg_3} displays both the expected (red line) and the actual (blue dots) node degrees of the graphs. The vertical dashed lines indicate the value of $i_0$. The relevance of that value is clearly visible: Nodes whose index is less than $i_0$ have comparable degrees, whereas the largest degree variation is produced by the remaining nodes. The point where the vertical line crosses the continuous red line sets up the scale of the largest hubs in the network.   
Figure \ref{fig:n_k_3} shows the degree profiles, that is, $n_k$ vs $k$ (blue dots), 
together with the expected power law (red line).
The abscissa of the vertical dashed line is $w_n$, which bounds from below the range where the degree distribution is expected to follow the power law. In fact, the number of nodes whose degree is smaller than $w_n$ departs 
ostensibly from the $\mathcal{O}(k^{-\gamma})$ behavior.
These small degrees are due to statistical fluctuations in nodes with very high indices, as the model specifies only the average values of the node degrees.

\begin{table}[t]   
\centering
\caption{Parameters for Figure \ref{fig:deg_3} and Figure \ref{fig:n_k_3}.}
\begin{tabular}{ccccc}
\toprule
$\gamma$ & $i_0$ & $w_n$ & $\mean(\b w)$ & avg.\ degree \\
\midrule
 $2.3$ & $25.1698$ & $2.3032$ & $7.4814$ & $7.5296$ \\
 $2.6$ & $13.4303$ & $3.7469$ & $9.1560$ & $9.1875$ \\
 $2.9$ & $5.5979$ & $4.7354$ & $9.7028$ & $9.7249$ \\
\bottomrule  
\end{tabular}
\end{table}

\begin{figure}[t]
    \centering
    \includegraphics[width=1\textwidth]{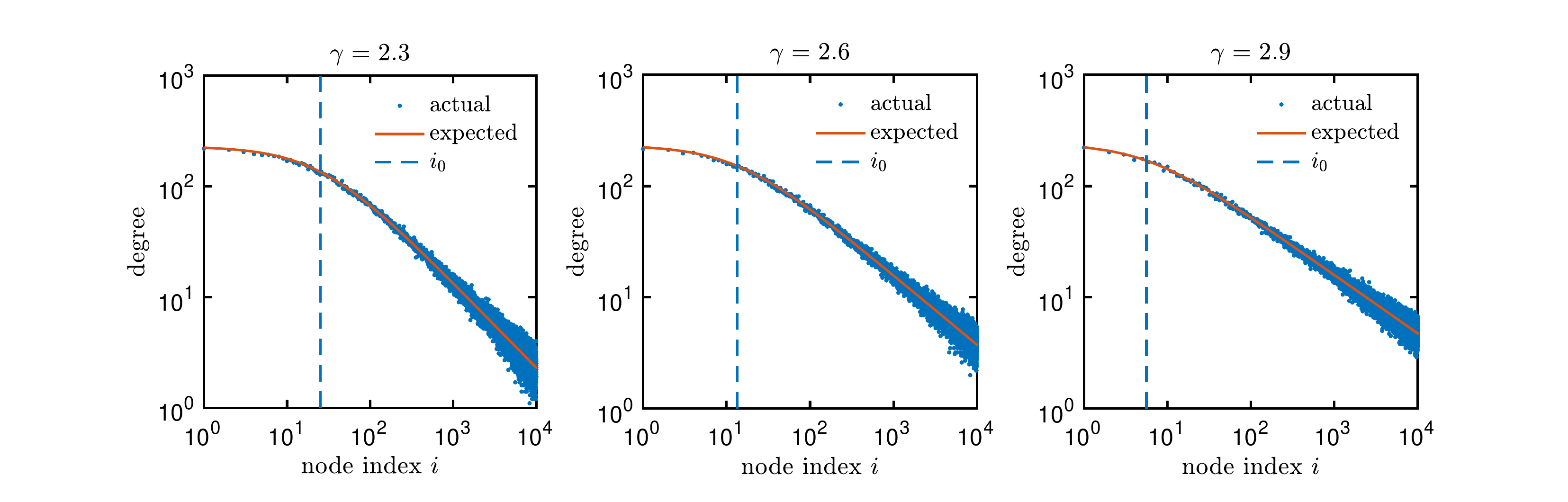}
     \caption{Actual vs expected node degrees of random graphs drawn from $G(\b w)$
     with $n = 10000$, $d = 10$, 
     $M = \sqrt{dn/2}$ and for three choices of $\gamma$.
     The values represented by blue dots are
     averages over 10 random instances. The vertical dashed lines show the value of $i_0$.}
    \label{fig:deg_3}
\end{figure}

\begin{figure}[t]
    \centering
    \includegraphics[width=1\textwidth]{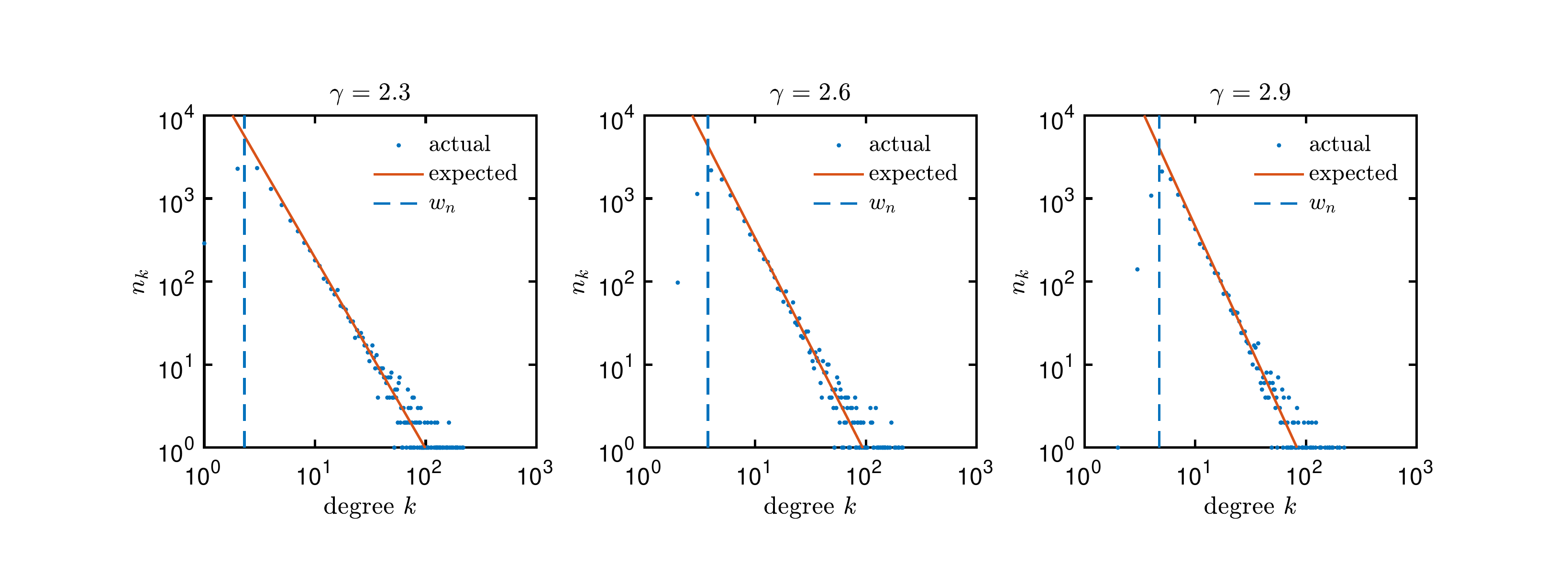}
     \caption{Actual vs expected degree profile $n_k$ of random graphs drawn from $G(\b w)$
     with $n = 10000$, $d = 10$, 
     $M = \sqrt{dn/2}$ and for three choices of $\gamma$.
     The values represented by blue dots are 
     averages over 10 random instances. The vertical dashed lines show the value of $w_n$.}
    \label{fig:n_k_3}
\end{figure}

\begin{figure}[t]
    \centering
    \includegraphics[width=1\textwidth]{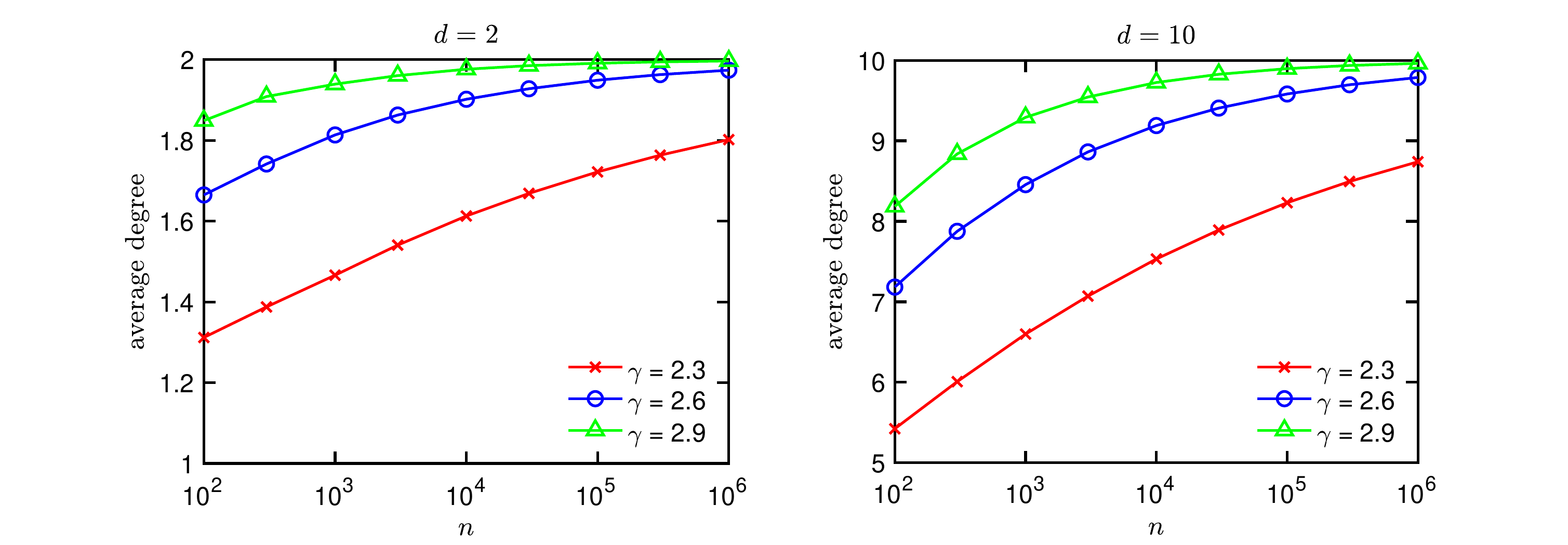} 
     \caption{Average degree in scale-free Chung--Lu random graphs
     of sizes growing from $10^2$ to $10^6$.
     Each color represent a different exponent for the power law: red crosses for $\gamma=2.3$, blue circles for $\gamma=2.6$ and green triangles for $\gamma=2.9$. Each point represents the average over $10$ different graphs
     with $d=2$
     (left) or $d = 10$ (right) and $M = \sqrt{dn/2}$.}
    \label{fig:average_degree}
\end{figure}

Finally, notice that, from equation \eqref{eq:averagedegree} it is clear that 
$\mean(\b w)$ converges towards $d(n)$ 
as $\mathcal{O}(\theta(n)^{1-p})$.
Then, if $p$ is approximately $1$, i.e., $\gamma$ is close to $2$, the convergence
can be very slow. 
This is illustrated in Figure \ref{fig:average_degree}, where we plot the average degree of different networks generated by the algorithm in Figure \ref{fig:code} as a function of the network dimension $n$. As we can see from the plot, as $n$ increases the average degree converges from below
to the parameter $d$ of the formula \eqref{eq:c_i_0}. 
The convergence is faster when $\gamma$ is
large and $d$ is small.

\begin{figure}[t]
    \centering
    \includegraphics[width=1\textwidth]{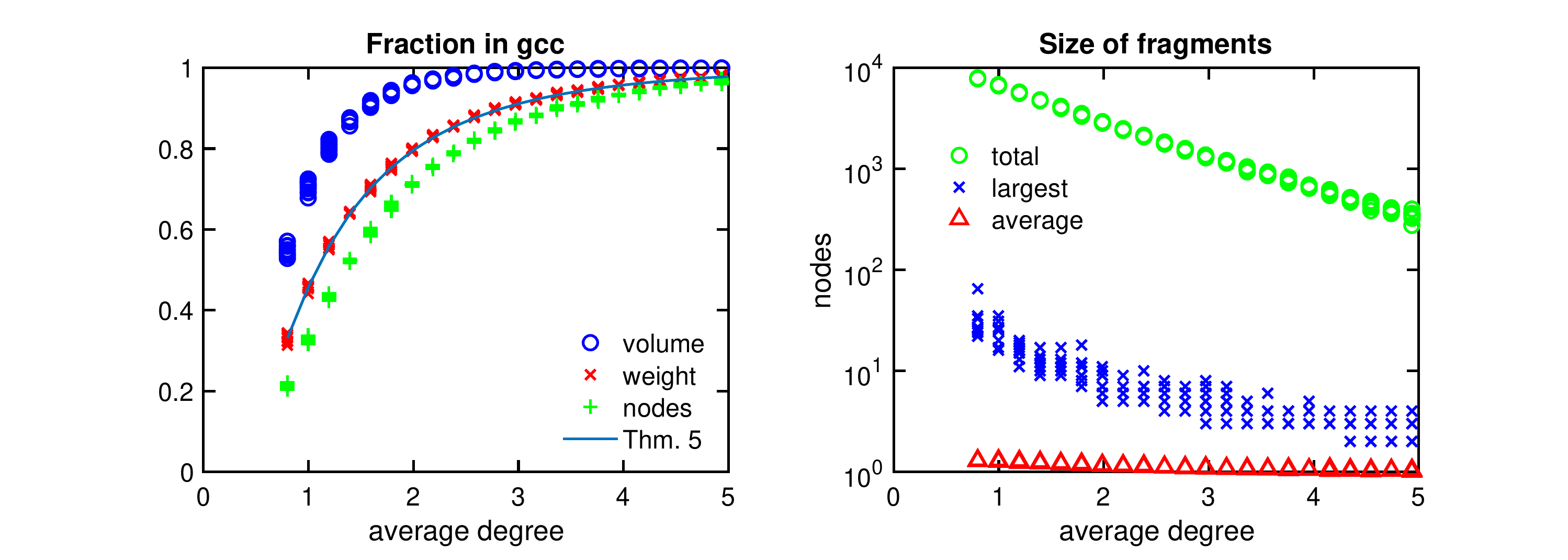}
     \caption{Statistics on the giant component (left) and the fragments (right) on scale-free Chung--Lu random graphs of size $n = 10000$ with $\gamma = 3$. The panel on the left shows the fraction of nodes $|S|/|V|$, volume $\sum_{i \in S}d_i/\sum_{i \in V}d_i$ and weight $\sum_{i \in S}w_i/\sum_{i \in V}w_i$ of the giant component $S$ as a function of the effective average degree $d\in \{0.8,1,1.2,1.4,...,4.8,5\}$.  The panel on the right shows the total number of nodes that are not in the giant component and the largest and average sizes of the connected components outside the giant component. }
    \label{fig:6}
\end{figure}

Finally, we present in Figure \ref{fig:6}
some statistics on the giant connected component (gcc) of scale-free Chung--Lu random graphs, to illustrate the content of Theorem \ref{thm:gcc}. 
For each $d = 0.8:0.2:5.0$ 
we generated 10 random graphs $G \in G(\b w)$ of size $n = 10000$ according to the construction in Corollary \ref{cor:M(n)}, with $\gamma = 3$ and $\mean(\b w) = d$. 
For each graph $G$ we computed the node subset $S$ comprising the giant component, measured the overall weight $\sum_{i\in S}w_i$ and the volume $\sum_{i\in S} d_i$, where $d_i$ is the effective degree of node $i$ in $G$.
The left panel of Figure \ref{fig:6} shows these numbers as fractions of the corresponding quantities in $G$, with respect to the effective average degree of $G$. The continuous line represents the value of $\lambda_0$ appearing in Theorem \ref{thm:gcc}.
Interestingly, while the actual overall weight of $S$ behaves precisely as expected by Theorem \ref{thm:gcc} and the effective average degree of $G$ closely mirrors $\mean(\b w)$,
the coefficient $\lambda_0$ overestimates the node fraction $|S|/n$ and largely underestimates the volume fraction $\sum_{i\in S}d_i/\sum_{i=1}^n d_i$. We can conclude that the giant component consists mainly of nodes whose effective degree $d_i$ is considerably larger than the corresponding weight $w_i$. 
This fact is rather exceptional. Indeed, the volume of an arbitrary subset of a Chung--Lu random graph is usually well approximated by its overall weight \cite{chunglubook,CL06}.

The right panel of Figure \ref{fig:6} displays data
about the fragments in $G$, that is, the connected subgraphs that remain after removing the giant component from $G$.
The total fragment size, that is,
the overall number of nodes that do not belong to $S$, is illustrated by the green circles. The blue crosses and the red triangles 
represent the largest and the average fragment size, respectively.
The picture shows that the average number of nodes outside the giant component decays exponentially as $d$ increases, and the fragments consist mainly of isolated nodes or very small-sized subgraphs.

\section{Conclusions}

We investigated the possibility of generating large random graphs having a power-law degree distribution using 
the Chung--Lu generative
model. This model is defined in terms of $n$ parameters, $w_1,\dots,w_n$,  where $n$ is the number of nodes, and is particularly relevant as it is the only random graph model that does not introduce correlations between the degrees of two nodes connected by an edge. 
Under the condition $\max_iw_i^2 \leq \sum_iw_i$, here called admissibility,
the parameters $w_i$ determine the node degrees, in expectation.

In this work, we analyzed the possibility of using admissible parameters $w_i$ to generate large networks having a
power-law degree distribution 
with arbitrary exponent $\gamma$.
Not surprisingly, the admissibility condition
imposes severe restrictions on the resulting degree sequence and, in some cases, also on the network size. 
In particular, we proved that, under appropriate 
hypotheses, 
the general formula $w_i = c(i_0+i)^{-p}$ can serve to generate 
arbitrarily large random graphs having a power-law degree distribution with exponent $\gamma = 1+1/p$,
at least within certain exponent and degree ranges.
Whilst our main focus is on the exponent range $2 < \gamma < 3$, 
which is the range 
often encountered in real-world networks,
we also considered rather general exponents.
Our main results provide explicit formulas for the coefficients $c$ and $i_0$, to fulfill some desirable requirements
on the connectivity of the network and the behavior of the largest and average degrees, in expectation. 

Furthermore, we proposed an algorithm to produce random graphs belonging to the Chung--Lu model, with arbitrary parameters $w_1,\dots,w_n$. The algorithm can be easily programmed in \matlab, Python,
and other high-level languages as well, without using explicit for-loops, which render it very fast. Together with the formulas  
alluded to above, this algorithm yields a flexible random graph generator for scale-free networks with general exponent $\gamma$. 
In practice, the observed running time for generating sparse networks up to million nodes is roughly proportional to the number of generated edges.

\subsection*{Funding}
The work of Dario Fasino and Arianna Tonetto has been carried out in the framework of the departmental research project ``ICON: Innovative Combinatorial Optimization in Networks'', Department of Mathematics, Computer Science and Physics (PRID 2017),
University of Udine, Italy. The work of Dario Fasino and Francesco Tudisco has been partly supported by 
Istituto Nazionale di Alta Matematica, INdAM-GNCS, Italy.

\end{document}